\documentclass[runningheads]{llncs}
\usepackage[T1]{fontenc}
\usepackage{graphicx}
\usepackage{amsmath, amsthm, amssymb}
\usepackage[colorlinks=true, allcolors=blue]{hyperref}
\usepackage{cleveref}
\usepackage{xspace}
\usepackage{todonotes}
\usepackage{comment}
\usepackage{algpseudocode}
\usepackage{enumitem}
\usepackage{booktabs,tabularx}
\usepackage{lipsum}

\theoremstyle{plain}

\theoremstyle{definition}

\newtheorem{alg}{Algorithm}
\newtheorem{prb}{Problem}

\newcommand{\homo}{\textsc{Hom}$_<$\xspace}
\newcommand{\homom}{\textsc{Hom}$_<^M$\xspace}

\newcommand{\NP}{\textsf{NP}\xspace}
\newcommand{\wone}{\textsf{W[1]}\xspace}

\newcommand{\Oh}{\mathcal{O}}

\newcommand{\core}{\textsc{CORE}$_<$\xspace}

\newcommand{\subg}{\textsc{SUB}$_<$}
\newcommand{\subgm}{\textsc{SUB}$_<^M$}
\newcommand{\colcore}{\textsc{ColCORE}$_<$}
\newcommand{\colhomo}{\textsc{ColHom}$_<$}

\newcommand{\PP}{\textsf{P}}

\newcommand{\problemStatement}[3]{%
  \begin{center}
  \begin{tabularx}{\columnwidth}{@{}lX@{}}
  \toprule
  \multicolumn{2}{@{}c@{}}{\textsc{#1}}\tabularnewline
  \midrule
  \bfseries Input:    & #2 \\
  \bfseries Question: & #3 \\
  \bottomrule
  \end{tabularx}
  \end{center}
}

\begin{document}
\title{On Computational Aspects of Ordered Matching Problems}
%
%
\author{Michal \v{C}ert\'{\i}k \inst{1}\orcidID{0009-0008-6880-4896} \and
Andreas Emil Feldmann \inst{2}\orcidID{0000-0001-6229-5332} \and
Jaroslav Ne\v{s}et\v{r}il \inst{1}\orcidID{0000-0002-5133-5586} \and
Pawe\l{} Rz\k{a}\.zewski \inst{3}\orcidID{0000-0001-7696-3848}\thanks{Supported by the National Science Centre grant 2024/54/E/ST6/00094.}}
\authorrunning{M. \v{C}ert\'{\i}k et al.}
%
\institute{Computer Science Institute, Faculty of Mathematics and Physics \\
Charles University\\
Prague, Czech Republic \and
Department of Computer Science \\
University of Sheffield\\
Sheffield, United Kingdom \and
Warsaw University of Technology\\
and University of Warsaw\\
Warsaw, Poland}
\maketitle              
\begin{abstract}

     Ordered matchings, defined as graphs with linearly ordered vertices, where each vertex is connected to exactly one edge, play a crucial role in the area of ordered graphs and their homomorphisms. Therefore, we consider related problems from the complexity point of view and determine their corresponding computational and parameterized complexities. We show that the subgraph of ordered matchings problem is \NP-complete and we prove that the problem of finding ordered homomorphisms between ordered matchings is \NP-complete as well, implying \NP-completeness of more generic problems. In parameterized complexity setting, we consider a natural choice of parameter - a number of vertices of the image ordered graph. We show that in contrast to the complexity context, finding homomorphisms if the image ordered graph is an ordered matching, this problem parameterized by the number of vertices of the image ordered graph is FPT, which is known to be W$[1]$-hard for the general problem. We also determine that the problem of core for ordered matchings is solvable in polynomial time which is again in contrast to the \NP-completeness of the general problem. We provide several algorithms and generalize some of these problems into ordered graphs with colored edges.

\keywords{Computational Complexity \and Parameterized Complexity \and Ordered Graphs  \and Homomorphisms  \and Ordered Matchings \and Ordered Core}
\end{abstract}

\section{Introduction}

An \emph{ordered graph} is a graph whose vertex set is totally ordered (see example in Figure~\ref{fig:OrdHomsInterval}).

For two ordered graphs $G$ and $H$, an \emph{ordered homomorphism} from $G$ to $H$ is a mapping $f$ from $V(G)$ to $V(H)$ that preserves edges and vertices ordering, that is,
\begin{enumerate}
    \item for every $uv \in E(G)$ we have $f(u)f(v) \in E(H)$,
    \item for $u,v \in V(G)$, if $u \leq v$, then $f(u) \leq f(v)$.
\end{enumerate}

Note that the second condition implies that the preimage of every vertex of $H$ forms a segment or an \emph{independent interval} in the ordering of $V(G)$ (see Figure~\ref{fig:OrdHomsInterval}).

\begin{figure}[ht]
\begin{center}
\includegraphics[scale=0.8]{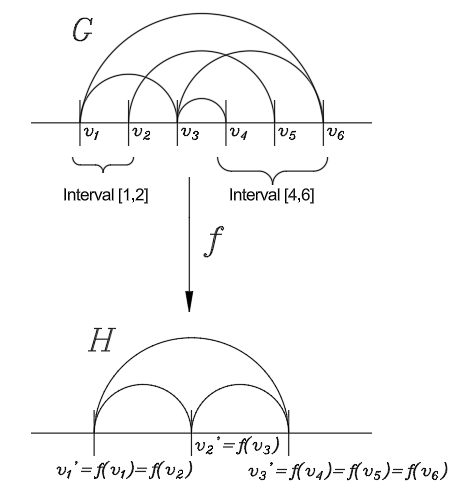}
\end{center}
\caption{Ordered Homomorphism $f$ and Independent Intervals.}
\label{fig:OrdHomsInterval}
\end{figure}

We define an \emph{Ordered Matching} as an ordered graph in which every vertex is incident to precisely one edge.

Let us continue with the definition of an \emph{Ordered Retraction} of $G$ to $H$ being an ordered homomorphism $f:G\to H$ such that $f(v)=v$ for all $v\in V(H)$.

If this ordered retraction of $G$ to $H$ exists, we say that $G$ \emph{order-retracts} to $H$ or that $H$ is an \emph{order-retract} of $G$. Notice that, as for digraphs (see \cite{HellNesetrilGraphHomomorphisms}), if ordered retraction exists, we have ordered homomorphism of $G$ to $H$ and inclusion-ordered homomorphism of $H$ to $G$, therefore $G$ and $H$ are (ordered) homomorphically equivalent.

Let us now define \emph{an Ordered Core} as an ordered graph that does not retract to a proper ordered subgraph. We show in ~\cite{certik_core_2025} that, as for unordered digraphs (see \cite{HellNesetrilGraphHomomorphisms}), an ordered graph $G$ is an ordered core if and only if there is no ordered homomorphism from $G$ to a proper ordered subgraph of $G$, and that every ordered graph is homomorphically equivalent to a unique ordered core.


Further on we may call an independent interval simply an interval, and an ordered graph, ordered homomorphism, ordered core, and ordered matching simply graph, homomorphism, core and matching, respectively.

\section{Motivation and Overview}

Ordered graphs naturally arise in several situations: Ramsey theory (~\cite{Nesetril1996},~\cite{Hedrlín1967},~\cite{balko2022offdiagonal}), extremal theory (~\cite{Pach2006},~\cite{conlon2016ordered}), category theory (~\cite{Hedrlín1967},~\cite{nesetril2016characterization}) and others. Recently, the notion of twin width of graphs has been shown to be equivalent to NIP (or dependent) classes of ordered graphs (~\cite{bonnet2021twinwidth}, ~\cite{bonnet2024twinwidth}) thus meeting graph theory and model theory.

Homomorphisms of ordered graphs confirm and complement the above line of research: they are more restrictive (than standard homomorphisms, see, e.g. ~\cite{HellNesetrilGraphHomomorphisms}) but display richness on its own (see, e.g. ~\cite{Axenovich2016ChromaticNO}, \cite{duffus1995computational-7a8}). In ~\cite{kun2025dichotomy-dd6} it has even recently been shown that ordering problems for graphs defined by finitely many forbidden ordered subgraphs capture the class \NP.

The significance of ordered matchings is known in the domain of ordered graphs and their homomorphisms (see, e.g., ~\cite{balko2022offdiagonal}, ~\cite{Balko_2020}, ~\cite{conlon2016ordered}, ~\cite{nescer2023duality}). Even within homomorphisms of ordered relational structures, analogues or ordered matchings play a critical role (see ~\cite{nescer2023dualityRel}). Therefore, we examine associated problems with the aim of determining their computational and parameterized complexities.

We start the article by showing the \NP-completeness of the ordered matching subgraph problem, in Section ~\ref{sec:Submatching}, which naturally implies the \NP-completeness of the more general question of the ordered graph subgraph problem.

In Section ~\ref{sec:homomatching}, we show the \NP-completeness of the question whether there exists an ordered homomorphism between given ordered matchings. This again implies the \NP-completeness of the more general question of finding ordered homomorphisms between given ordered graphs.

Section ~\ref{sect:ParamCompMatching} focuses on parameterized complexity of problems related to ordered homomorphisms $G\to H$ between ordered graph $G$ and ordered matching $H$, parameterized by $|V(H)|$. We show that when parameterized by $|V(H)|$, these problems are easy. We show in ~\cite{certik_complexity_2025} that this contrasts with the general choice of ordered graphs, where the parameterized version of this problem by $|V(H)|$ is hard.

We also show in Section ~\ref{sect:MatchCore} that if the core of an ordered graph $G$ is an ordered matching, then the question whether $G$ is a core is easy. On the other hand, we show in ~\cite{certik_core_2025} that the decision whether a general ordered graph is a core is \NP-complete.

The article ends with several open problems and outlines of related research.

\section{Subgraphs of Ordered Matchings}
\label{sec:Submatching}

We start with a definition of the problem of determining whether an input ordered matching $G$ is a subgraph of an input ordered matching $H$.

\begin{prb}
\end{prb}
\problemStatement{\subgm}
  {Ordered matchings $G$ and $H$.}
  {Is $G$ a subgraph (resp. an induced subgraph) of $H$?}

A problem concerning ordered graphs and their subgraphs, albeit with different specifications and implications, is addressed in~\cite{duffus1995computational-7a8}.

We show that the problem \subgm is \NP-complete even for the following specific class of ordered matchings.

For an ordered graph $G$ with vertices ordered as follows $V(G)=(v_1,\ldots,v_n)$, an $i$\emph{-th cut} $C = (S, T)$ is a partition of $V(G)$ of a graph $G$ into two subsets $S=(v_1,\dots,v_i)$ and $T=(v_{i+1},\dots,v_n)$. Then we say that an edge $e$ \emph{goes across} the $i$-th cut $C = (S, T)$ if and only if one vertex of the edge $e$ belongs to $S$ and another endpoint of the edge $e$ belongs to $T$.

An ordered graph $G$ is \emph{separated} if there exists $i \in [|V(G)|-1]$ such that all edges go across the $i$-th cut.

\begin{theorem}\label{thm:matsub}
    Given two separated ordered matchings $M$ and $N$, it is \NP-complete to decide whether $N$ is a subgraph (resp., an induced subgraph) of $M$.
\end{theorem}

\begin{proof}
    We reduce from the \textsc{Permutation Pattern Matching Problem}. There we are given two permutations $\pi : [n] \to [n]$ and $\Pi : [m] \to [m]$, where $m \geq n$,
    and we ask whether $\Pi$ contains $\pi$ as a sub-permutation (i.e., whether there are $n$ elements whose relative positions in $\Pi$ are as in $\pi$). As shown in ~\cite{BoseBussLubiw1993}, this problem is \NP-complete.

    Let $M$ be the $m$-edge matching, consisting of the edges $(i, m+\Pi(i))$ for $i \in [m]$.
    Similarly, let $N$ be the $n$-edge matching, consisting of the edges $(i, n+\pi(i))$ for $i \in [n]$ (see Figure ~\ref{pic:MatchingsSub}).
    We claim that $N$ is a subgraph of $M$ if and only if $\Pi$ contains $\pi$.

    \begin{figure}[p]\centering
    \includegraphics[width=\textwidth,height=\textheight,keepaspectratio]{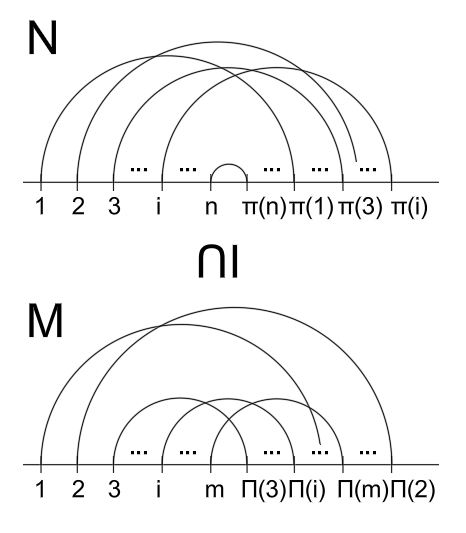}
    \caption{Ordered Graphs $N$ and $M$ from Theorem ~\ref{thm:matsub}.}
    \label{pic:MatchingsSub}
    \end{figure}

    First, suppose that $\Pi$ contains $\pi$, that is, there are $1 \leq i_1 < \ldots <i_m \leq n$ such that $\Pi$ restricted to $(i_1,\ldots,i_n)$ is $\pi$.
    We claim that the subgraph of $M$ induced by $X=\bigcup_{j \in [n]} \{ i_j, \Pi(i_j) \}$ is isomorphic to $N$.
    For contradiction, suppose otherwise.
    Clearly, $N[X]$ is a separated $n$-edge matching. In addition, the right endpoints of the edges are ordered according to $\pi$. Thus, $M[X]$ is indeed a copy of $N$.

    Now suppose that $N$ is a subgraph of $M$. Note that since both $M$ and $N$ are matchings, this subgraph must be induced. Let $X \in [2n]$ be such that $M[X]$ is isomorphic to $N$. Then $\pi$, defined by the edges $(j, n+\pi(j)), j \in [n]$ of $M[X]$, is clearly a subpermutation of $\Pi$, defined by the edges $(i, m+\Pi(i)), i \in [m]$ of $M$, as $\Pi$ restricted to $M[X]$ is $\pi$.
    
\end{proof}

Let us now define a more generic problem of determining whether an ordered graph $G$ is a subgraph of an ordered graph $H$.

\begin{prb}
\end{prb}
\problemStatement{\subg}
  {Ordered graphs $G$ and $H$.}
  {Is $G$ a subgraph (resp. an induced subgraph) of $H$?}

 Theorem ~\ref{thm:matsub}, of course, implies the following.

\begin{corollary}
    \subg ~and \subgm ~are \NP-complete.
\end{corollary}

\begin{proof}
    Since both problems \subg ~and \subgm are more general than the problem in Theorem ~\ref{thm:matsub} and \subg ~and \subgm ~are, of course, in \NP, the statement follows.
\end{proof}

\section{Finding Homomorphisms of Ordered Matchings}
\label{sec:homomatching}

It is easy to see that for a fixed ordered graph $H$, the (ordered) $H$-coloring problem, defined as the problem of finding an ordered homomorphism from a given ordered graph $G$ to a fixed ordered graph $H$, is in \PP ~(see, e.g. ~\cite{nescer2023duality}, ~\cite{certik_complexity_2025}).

In this section, we will therefore focus on the \homo computational problem, where $H$ is not fixed and is part of an input. Let us start with a definition of \homo for ordered matchings, whose input is a pair of ordered matchings $G$ and $H$, and we ask if $G$ admits an ordered homomorphism to $H$.

\begin{prb}{\homom}
\label{Prb:Hom<}
\end{prb}
\problemStatement{\homom}
  {Ordered matchings $G$ and $H$.}
  {Does there exist an ordered homomorphism $f:G\to H$?}

We now show that the \homom problem is \NP-complete.

\begin{theorem}
\label{thm:MatchingsHomNPC}
    Given two ordered matchings $M$ and $N$, it is \NP-complete to decide whether $N \to M$.
\end{theorem}
\label{thm:matchNPC}

\begin{proof}
    We again reduce from the \textsc{Permutation Pattern Matching Problem}.

    Let $M$ consist of the $m$-edge matching, consisting of the edges $(i, m+\Pi(i))$ for $i \in [m]$ (as in the proof of Theorem ~\ref{thm:matsub}) and the following edge set. Take the $2m$ vertices $i$, ordered naturally, created by the previous $m$-edge matching and add the vertices $a_i$ and $b_i, a_i<b_i$ between the vertices $i$ and $i+1, i\in [2m-1]$. Add vertices $c_m$ and $d_m, c_m<d_m$ between vertices $b_m$ and $m+1$. The positions of the vertices $a_i$ and $b_i, i\in [2m-1]$ and $c_m$ and $d_m$ are uniquely defined. Create an edge between the vertices $a_i$ and $b_i, i\in [2m-1]$ and create an edge between the vertices $c_m$ and $d_m$. The ordered graph $M$ is an ordered matching with $3m$ edges (see example of an ordered graph $M$ in Figure ~\ref{pic:MatchingsHom}).
    
    Similarly, let $N$ consist of the $n$-edge matching, consisting of the edges $(i, n+\pi(i))$ for $i \in [n]$ and the following edge set. Take the $2n$ vertices $i$, ordered naturally, created by the previous $n$-edge matching and add the vertices $a_i$ and $b_i$ between the vertices $i$ and $i+1, i\in [2n-1]$. Add vertices $c_n$ and $d_n$ between vertices $b_n$ and $n+1$. Create an edge between the vertices $a_i$ and $b_i, i\in [2n-1]$ and create an edge between the vertices $c_n$ and $d_n$ (see example of an ordered graph $N$ in Figure ~\ref{pic:MatchingsHom}).

    \begin{figure}[p]\centering
    \includegraphics[width=\textwidth,height=\textheight,keepaspectratio]{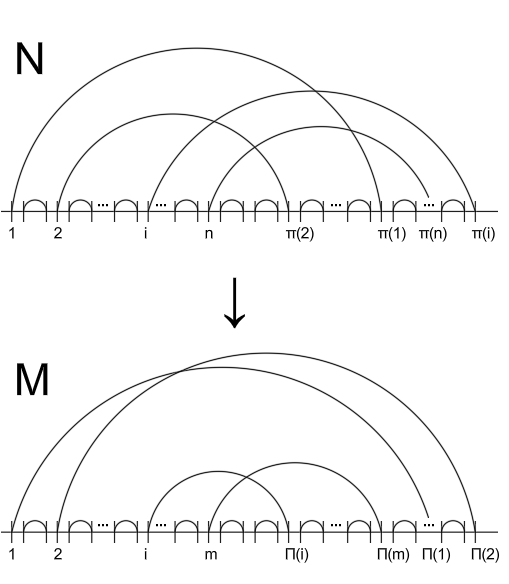}
    \caption{Ordered Homomorphism from $N$ to $M$ from the proof of Theorem ~\ref{thm:MatchingsHomNPC}}
    \label{pic:MatchingsHom}
    \end{figure}
    
    We claim that there exists an ordered homomorphism $N\to M$ if and only if $\Pi$ contains $\pi$.
    Denote the edges $(i, n+\pi(i))$ for $i \in [n]$ as \emph{permutation edges} and edges in between $a_i$ and $b_i, i\in [2n-1]$ and $c_n$ and $d_n$ by \emph{auxiliary edges}. We adopt the same notation for the edges in $M$.

    First, suppose that $\Pi$ contains $\pi$, that is, there are $1 \leq i_1 < \ldots <i_m \leq n$ such that $\Pi$ restricted to $(i_1,\ldots,i_m)$ is $\pi$. From Theorem ~\ref{thm:matsub} we know that the subgraph of $M$ induced by permutation edges of $M$ must contain a graph $M[X]$ isomorphic to a subgraph of $N$ induced by permutation edges of $N$. Auxiliary edges in $N$ ensure that this mapping must be injective.
    
    For contradiction, suppose otherwise. Suppose that two permutation edges of $N$ map to one permutation edge of $M$. This cannot happen, as there is always an auxiliary edge in between the ends of the permutation edges. Now suppose that a permutation edge $e$ in $N$ maps to an auxiliary edge $e'$ in $M$. This cannot happen either, as there are always at least two auxiliary edges between the ends of $e$ in $N$ and there are no edges between the ends of an auxiliary edge $e'$ in $M$. Therefore, the permutation edges of $N$ map injectively to an induced subgraph $M[X]$ of $M$. 
    
    It remains to show that the auxiliary edges of $N$ can also always be mapped to the edges of $M$. But this is also obvious, as there are always auxiliary edges between the permutation edges $M[X]$ in $M$.

    Now suppose that there exists an ordered homomorphism $N\to M$. We showed that the mapping of the permutation edges of $N$ must be injective to the permutation edges of $M$. Therefore, this mapping defines an induced subgraph $M[X]$ of $M$, where $M[X]$ is isomorphic to $N$. But then we follow the same argument as in the proof of Theorem ~\ref{thm:matsub}, showing that $\Pi$ restricted to $M[X]$ is a subpermutation $\pi$, $\Pi$ and $\pi$ defined by permutation edges of $M$ and $N$, respectively.
    
\end{proof}

Let us now define the more generic problem \homo, where the inputs are any ordered graphs.

\begin{prb}{\homo}
\label{Prb:Hom<}
\end{prb}
\problemStatement{\homo}
  {Ordered graphs $G$ and $H$.}
  {Does there exist an ordered homomorphism $f:G\to H$?}

The following result then naturally follows from Theorem ~\ref{thm:matchNPC}.

\begin{corollary}
    \homo is \NP-complete.
\end{corollary}

\begin{proof}
    Since \homo problem is more general than the \homom in Theorem ~\ref{thm:matchNPC} and \homo ~is, of course, in \NP, the statement follows.
\end{proof}

\section{Parameterized Complexity of Homomorphisms of Ordered Matchings}
\label{sect:ParamCompMatching}

Let us show that the \homo problem, when parameterized by $|V(M)|$ for an ordered matching $M$, is FPT.

\begin{theorem}
\label{thm:matchfpt}
    Let $G$ be an ordered graph and $M$ be an ordered matching. Then deciding whether there exists an ordered homomorphism $G\to M$ is fixed-parameter tractable with respect to $|V(M)|$.
\end{theorem}

\begin{proof}
    
We start by constructing the smallest possible ordered matching $G'$ out of $G$, by mapping the subgraphs of $G$ into disjoint edges.

    Let $G$ have $g$ vertices and let $M$ have $m$ vertices. We denote these vertices in $G$ by $1,2,\ldots,g$.

    Let us construct the smallest possible ordered graph $G'$, in the number of edges $g'$ in $G'$.

    Take an interval of independent vertices $i,i+1,\ldots,j$ in $G$, which have their neighbors only among the vertices $j+1,j+2,\ldots,g$. We can determine whether edges incident with the vertices $i,i+1,\ldots,j$ can be mapped to only one edge in $G$ as follows. Let us denote a set of vertices connected to vertices $i,i+1,\ldots,j$ as $V_{ij}$. Then we can map the edges incident with the vertices $i,i+1,\ldots,j$ to one edge in $G$ if and only if the set of vertices $V_{ij}$ form an interval in the ordering of $G$. The check whether the vertices in $V_{ij}$ form an interval in the ordering of $G$ can, of course, be performed in polynomial time.
    
    Now we define the following greedy algorithm.

\begin{alg}

\hfill

\label{alg:G'}
\begin{itemize}
    \item[] \textbf{Input:} An ordered graph $G$.
    \item[] \textbf{Output:} The smallest ordered matching $G'$, such that $G\to G'$. 
    \item[] \textbf{Method:} Iterate on the vertices in the given order and map the ordered subgraphs of $G$ to an edge whenever possible.
\end{itemize}


\begin{enumerate}
\item Set $A=\emptyset$ and $i=1$.
\item Add $i$ to $A$ and set $a_{max}=i$
\item For $j=i+1$ to $g$
    \begin{enumerate}
    [label*=\arabic*.,leftmargin=3\parindent]
        \item Add $j$ to $A$
        \item If any of the vertices in $A$ are connected or if $j$ is connected to any of the vertices $1,2,\ldots,i-1$, exit the loop
        \item If the edges incident with vertices in $A$ can be mapped to one edge in $G$, set $a_{max}=j$
    \end{enumerate}
\item Map all the edges incident with the vertices $i,i+1,\ldots,a_{max}$ to one of these edges and set $i=a_{max}$
\item For $j=i+1$ to $g$
    \begin{enumerate}[label*=\arabic*.,leftmargin=2\parindent]
        \item if $j$ is connected only to vertices among $j+1,j+2,\ldots,g$ then set $i=j, A=\emptyset$ and go to step 2.
        \item If $j=g$ stop
    \end{enumerate}
\end{enumerate}
\end{alg}

We note that if there exists an ordered homomorphism $G\to M$, then $G$ can be partitioned into disjoint subgraphs that each map to an edge. We also notice that if the ordered subgraph $G'$ of $G$ resulting from the greedy Algorithm~\ref{alg:G'} above has the minimum number of edges (minimum with respect to a number of edges of other ordered matching subgraphs of $G$, where $G$ could map) and is ordered matching, then it must be an order-retract of $G$ as if every subgraph of $G$ maps to an edge, then this edge must be a part of that subgraph. Here we note that this is the reason this approach works only for the matchings (since if e.g. two disjoint edges map into a path of length two, this might not be an ordered subgraph of the original ordered graph). We also note that if $G'$ is an ordered matching and is minimum with respect to the number of edges (again minimum with respect to the number of edges of other ordered matching subgraphs of $G$, where $G$ could map), then $G'$ is the core of $G$.

Therefore, if $G$ is not a core or $G$ is not an ordered matching and there exists an ordered homomorphism $G\to M$, then in the ordered retraction $f:G\to G$, there exists at least one subgraph in $G$, containing at least three vertices, that maps to a single edge in $G$. In addition, if there are two or more subgraphs in $G$ that map to an edge in $G$, these mappings can be performed consecutively and independently of each other.

Denote a subgraph of $G$ by $G_{e}$. We note that the only way that the $G_{e}$ in $G$ can map to an independent edge in $G'$ is if and only if $G_e$ can be partitioned into two sets of independent vertices $i,i+1,\ldots,j$ and $V_{ij}$, where both of these sets are independent intervals and $i,i+1,\ldots,j$ is only connected to $V_{ij}$, and vice versa. The subgraphs $G_{e}$ in $G$ that can map to a single edge in $G$ can therefore be uniquely determined by the set $i,i+1,\ldots,j$.

We see that if there exists a subgraph $G_{e}$ with more than two vertices in $G$ that can be mapped to a single edge, the algorithm finds it and performs the mapping (since these subgraphs can be mapped to edges independently). Therefore, if the core of $G$ is an ordered matching $G'$, the resulting ordered graph $G'$ at the conclusion of the algorithm is equal to $G'$ (of course, minimum with respect to $g'$). The algorithm again runs in $\Oh(g^{\Oh(1)})$ steps.

Now we can compare the number of disjoint edges $g'$ in $G'$ and the number $h$ of disjoint edges in $M$.

If $g'>h$, then we know that there is no ordered homomorphism $G'\to M$, as $G'$ is minimal with respect to $g'$. Then, of course, there is no ordered homomorphism $G\to M$. 

If $g'\le h$, then we know that we can find the ordered homomorphism $G'\to M$ by examining all $2^h$ possible choices of edges in $M$ and seeing if any of them is isomorphic to $G'$ (the isomorphism check in the case of ordered graphs is, of course, in \PP). If we find such a match, we know that there is an ordered homomorphism $G'\to M$ and therefore also $G\to M$ (by transitivity). If none of the combinations gives us a graph isomorphic to $G'$, then we know that there does not exist an ordered homomorphism $G'\to M$ (here we also use the fact that $G'$ is minimal with respect to $g'$ and $G'$ is a core). We notice that we can also perform $2^m$ possible choices of vertices of $M$ (to derive the complexity stated in the theorem), knowing that if we do not consider one vertex in an edge in $M$, we might not consider the whole edge, since $G'$ is an ordered matching.

As the whole procedure takes at most $\Oh(2^mg^{\Oh(1)}g'^{\Oh(1)})$ steps, the result follows.

\end{proof}

Let us now define a \emph{colored ordered graph} as an ordered graph with colored edges and \emph{colored ordered homomorphism} as an ordered homomorphism that preserves the colors of edges.

We define in the same way a colored ordered homomorphism generalization of the \homo problem and determine the parameterized complexity of this problem parameterized by $|V(H)|$ for colored ordered matchings.

\begin{prb}
\end{prb}
\problemStatement{\colhomo}
  {Colored ordered graphs $G$ and $H$.}
  {Does there exist a colored ordered homomorphism from $G$ to $H$?}

\begin{corollary}
\label{cor:colhomommFPT}
    \colhomo ~parameterized by $|V(H)|$, where $G$ is a colored ordered graph and $H$ is a colored ordered matching, is fixed-parameter tractable.
\end{corollary}

\begin{proof}
    Again, we notice that the only change with respect to the proof of Theorem ~\ref{thm:matchfpt} is that the intervals of vertices need to be incident with the edges of the same color if these edges are to map to an edge.

    We therefore simply adjust Step 3.3 in the algorithm by adding the additional condition asking if the set $A$ contains only the vertices that are incident with edges of the same color. The result then follows from the proof of Theorem ~\ref{thm:matchfpt}.
\end{proof}

We will now generalize the notion of ordered matching. Let $H_0$ be an ordered graph. We define an \emph{ordered $H_0$-factor} as an ordered graph $H$ that is a disjoint union of $H_0$. For example, taking $H_0=K_2$, the resulting ordered $H_0$-factor is an ordered matching. In the following, we may denote ordered $H_0$-factor simply as $H_0$-factor.

The following then determines the parameterized complexity of \homo problem parameterized by $|V(H)|$, where $H$ is $H_0$-factor.

\begin{theorem}\label{thm:H0FPT}
    Let $H_0$ be a fixed connected ordered core and $G$ be an ordered graph. Then deciding whether there exists ordered graph $H$, such that $H$ is a disjoint unions of the graphs $H_0$, and $G$ admits an ordered homomorphism to $H$ is fixed-parameter tractable with respect to $|V(H)|$.
\end{theorem}

\begin{proof}

    Similarly to the proof of Theorem ~\ref{thm:matchfpt}, we again start by constructing the smallest possible ordered graph $G'$ out of $G$, for which $G\to G'$. Of course, $G'$ will also consist of disjoint unions of $H_0$.

    Let $g$ and $h$ be a number of copies of ordered graphs $H_0$ in $G$ and $M_h=H$, respectively, and $h_0=|H_0|$. Then $G$ has $m=gh_0$ vertices and $H$ has $n=hh_0$ vertices. We denote the vertices in $G$ by $1,2,\ldots,m$.

    Let us construct the smallest possible ordered graph $G'$ out of $G$, $G'$ smallest in the number $g'$ of disjoint copies of $H_0$. We proceed in a manner similar to that of Theorem~\ref{thm:matchfpt}.
    
    Take an interval of independent vertices $i,i+1,\ldots,j$ in $G$, which are all the first vertices of their copies of $H_0$. We know that each of the vertices $i,i+1,\ldots,j$ is part of exactly one $H_0$ copy (from the condition that all the vertices $i,i+1,\ldots,j$ in $G$ are the first vertices of their copies of $H_0$). Determining whether a vertex is the first vertex of its copy of $H_0$ is, of course, polynomial (in fact, linear in $m$).
    We can determine if these $j-i+1$ copies of $H_0$ incident with vertices $i,i+1,\ldots,j$ can be mapped to only one $H_0$ in $G$ as follows. Let us denote the sets of vertices $A_{k}=\{v^k_i,v^k_{i+1},\ldots,v^k_{j}\}, k=2,3,\ldots,h_0$ as the set of vertices corresponding to the $k$-th vertices of copies of $H_0$, incident with the vertices $i,i+1,\ldots,j$. Let $A_{1}=\{i,i+1,\ldots,j\}$. Notice that $A_{k}, k\in [h_0]$ are well-defined as $H_0$ is connected. Then we can map copies of $H_0$ incident with vertices $i,i+1,\ldots,j$ to one $H_0$ in $G$ if and only if each of the sets of vertices $A_{k}, k=2,3,\ldots,h_0$ forms an interval in the ordering of $G$. The check whether all edge sets $A_{k}, k=2,3,\ldots,h_0$ form an interval in the ordering of $G$ can be, of course, performed in polynomial time.
    
    Now we define the following greedy algorithm.

\begin{alg}(\homo $H_0$-Matchings Algorithm)

\hfill

\label{alg:G'H0}
\begin{itemize}
    \item[] \textbf{Input:} An ordered $H_0$-matching $G$.
    \item[] \textbf{Output:} The smallest (with respect to $g'$) ordered $H_0$-matching $G'$, such that $G\to G'$. 
    \item[] \textbf{Method:} Iterate on the vertices in the given order and map the ordered subgraphs of $G$ to $H_0$ whenever possible.
\end{itemize}


\begin{enumerate}
\item Set $A=\emptyset$ and $i=1$.
\item Add $i$ to $A$ and set $a_{max}=i$
\item For $j=i+1$ to $m$
    \begin{enumerate}
    [label*=\arabic*.,leftmargin=3\parindent]
        \item Add $j$ to $A$
        \item If $j$ is not the first vertex of $H_0$, exit the loop
        \item If the copies of $H_0$ incident with vertices in $A$ can be mapped to one copy of $H_0$ in $G$, set $a_{max}=j$
    \end{enumerate}
\item Map all the copies of $H_0$ incident with the vertices $i,i+1,\ldots,a_{max}$ to one of these $H_0$ copies and set $i=a_{max}$
\item For $j=i+1$ to $m$
    \begin{enumerate}
    [label*=\arabic*.,leftmargin=2\parindent]
        \item if $j$ is the first vertex of $H_0$ then set $i=j, A=\emptyset$ and go to step 2.
        \item If $j=m$ stop
    \end{enumerate}
\end{enumerate}
\end{alg}

We note that since $G$ is a $H_0$-matching, if $G$ is not a core, then $G$ must map to its ordered subgraph $G'$, which is also a $H_0$-matching (here, we need the property of $H_0$ being an ordered core).
Therefore, if $G$ is not a core, then in the ordered retraction $f:G\to G$ that is not onto, there exists at least one set of subgraphs $H_0$ in $G$ that maps to a single graph $H_0$ in $G$. Again, in addition, if there are two or more subgraphs $H_0$ in $G$ that map to different subgraphs $H_0$ in $G$, these mappings can be performed consecutively and independently of each other.

Let us denote by $G_{H_0}$ an ordered graph induced by a subset of vertices in $G$ induced by a subset of copies of $H_0$ in $G$. As observed before, the only way that $G_{H_0}$ in $G$, $G_{H_0}$ containing at least two copies of $H_0$, can map to a single $H_0$ in $G$ is if and only if each set of vertices $A_{k}, k\in [h_0]$ forms an interval. The set $G_{H_0}$ of $H_0$ in $G$ that can map to a single $H_0$ in $G$ can therefore be uniquely determined by the set $A_{1}$ (we again use that $H_0$ is connected).

We see that the greedy algorithm above therefore searches through all the possible $G_{H_0}$. Thus, if there exists a set $G_{H_0}$ with more than one copy of $H_0$ in $G$ that can map to a single $H_0$ in $G$, the algorithm finds it and performs the mapping. The resulting ordered graph $G'$ at the end of the algorithm is therefore minimum with respect to $g'$ - a number of disjoint copies of $H_0$ in $G'$. The algorithm runs in $\Oh(m^{\Oh(1)})$ steps.

Now we can compare the number $g'$ of disjoint copies of $H_0$ in $G'$ and the number $h$ of disjoint copies of $H_0$ in $H$. 

If $g'>h$, then we know that there is no ordered homomorphism $G'\to H$, as $G'$ is minimal with respect to $g'$ and $G'$ is a core. Then there is, of course, no ordered homomorphism $G\to H$ (as $G$ has more copies of $H_0$ than $G'$). 

If $g'\le h$, then we know that we can find the ordered homomorphism $G'\to H$ by examining all $2^h$ possible choices of disjoint unions of $H_0$ in $H$ (also $h^{g'}$ choices would be enough, since it is enough to choose $g'$ out of $h$ copies of $H_0$ in $H$) and seeing if any of them is isomorphic to $G'$ (again, the isomorphism check for ordered graphs is in \PP). If we find such a match, we know that there is an ordered homomorphism $G'\to H$ and therefore also $G\to H$ (by transitivity). If none of the combinations gives us a graph isomorphic to $G'$, then we know that there does not exist an ordered homomorphism $G'\to H$ (here we again use the fact that $G'$ is minimal with respect to $g'$ and $G'$ is a core). 

We again notice that we can also perform $2^n$ possible choices of vertices of $H$ 
, knowing that if we do not consider at least one vertex in a copy of $H_0$ in $H$, we might not consider the whole $H_0$, as $H_0$ is a core and $G'$ contains only copies of $H_0$.

As the whole procedure takes at most $\Oh(2^hm^{\Oh(1)}g'^{\Oh(1)})$ steps 
, the result follows.

\end{proof}

\section{Matchings Core Problems}
\label{sect:MatchCore}

Let us define the following problem, which we denote \core.

\begin{prb}
\end{prb}
\problemStatement{\core}
  {Ordered graph $G$.}
  {Is there a non-surjective ordered homomorphism $G\to G$?}

The following corollary then follows directly from the proof of Theorem ~\ref{thm:matchfpt} for the special case of \core, where the core of an input ordered graph $G$ is a matching.

\begin{corollary}
    Let $G$ be an ordered graph. Then, if an ordered matching $M$ is the core of $G$, finding $M$ is in \PP.
\end{corollary}


Note that if an ordered graph $G$ is an ordered matching, then finding its core in polynomial time is also an easy consequence since its core must be a matching.

The following is then a related colored ordered homomorphism problem (generalization of \core).

\begin{prb}
\end{prb}
\problemStatement{\colcore}
  {Colored ordered graph $G$.}
  {Is there a non-surjective colored ordered homomorphism $G\to G$?}

\begin{corollary}

    Let $G$ be an ordered colored graph. Then, if an ordered colored matching $M$ is the core of $G$, finding $M$ is in \PP.
\end{corollary}

\begin{proof}
The statement follows from the proof of Theorem ~\ref{thm:matchfpt} and Corollary ~\ref{cor:colhomommFPT}, where we do the same type of adjustment to Algorithm ~\ref{alg:G'} as we did in the proof of Corollary ~\ref{cor:colhomommFPT}.


\end{proof}

These results are in contrast with the \NP-completeness of the general \core problem (and therefore also \colcore), which we show in the article ~\cite{certik_core_2025}.

We also show in ~\cite{certik_complexity_2025} that the parameterized \homo problem, where input graphs are any ordered graphs, is \wone -hard. We observe that while the ordered graphs $G$ and $H$ in ~\cite{certik_complexity_2025} can be quite complex, the ordered graphs $G$ and $H$ in Theorem ~\ref{thm:H0FPT} are rather restricted and simple.

Therefore, it remains to fill the gap between these classes of ordered graphs and categorize which of the classes of ordered graphs makes \homo parameterized by $|V(H)|$ in \wone-hard and which of them in FPT. We will not discuss this question in the present article.





\begin{credits}

\subsubsection{\discintname}

\end{credits}
%
%
%
\bibliographystyle{splncs04}
\bibliography{mybibliography}
\end{document}